\documentclass[journal]{IEEEtran}
\usepackage{amsmath,amssymb,amsfonts,dsfont, amsthm}
\usepackage{bm}
\usepackage{algorithm, algorithmicx}
\usepackage{algpseudocode}
\usepackage{subcaption}
\def\BibTeX{{\rm B\kern-.05em{\sc i\kern-.025em b}\kern-.08em
		T\kern-.1667em\lower.7ex\hbox{E}\kern-.125emX}}
\usepackage{makecell}
\usepackage{breqn}
\usepackage{tikz}
\usepackage[english]{babel}
\usepackage{multirow}
\usepackage{comment}
\usepackage{upgreek}
\usepackage{textcomp}
\usepackage{hyperref}
\newtheorem{theorem}{Theorem}[section]

\setcellgapes{3pt}

\begin{document}

\title{Distributed Multi-Sensor Control for Multi-Target Tracking Using Adaptive Complementary Fusion for LMB Densities}

\author{Aidan~Blair, Amirali~Khodadadian~Gostar, Alireza~Bab-Hadiashar, Xiaodong~Li, and~Reza~Hoseinnezhad
\thanks{Submitted \today.}
\thanks{All authors are with RMIT University, Melbourne, Victoria (email: aidan.blair@rmit.edu.au).}}

\maketitle


\begin{abstract}
Tracking multiple targets in dynamic environments using distributed sensor networks is a fundamental problem in statistical signal processing. In such scenarios, the network of mobile sensors must coordinate their actions to accurately estimate the locations and trajectories of multiple targets, balancing limited computation and communication resources with multi-target tracking accuracy. Multi-sensor control methods can improve the performance of these networks by enabling efficient utilization of resources and enhancing the accuracy of the estimated target states. This paper proposes a novel multi-sensor control method that utilizes multi-agent coordinate descent to address this problem, ensuring distributed consensus of optimal sensor actions throughout the sensor network. To achieve this, a novel adaptive complementary fusion approach that prioritizes information from the most informative sensors is developed. Our method improves computational tractability and enables fully distributed control, ensuring the scalability and flexibility necessary for large-scale real-time sensing systems. Experimental results on several challenging multi-target tracking scenarios demonstrate that our approach significantly improves both multi-target tracking accuracy and computation efficiency over competing methods.
\end{abstract}


\section{Introduction}
\label{sec:Introduction}
This article addresses the problem of distributed multi-sensor control for multi-target tracking, a central challenge in statistical signal processing for large-scale sensing systems. In distributed sensor networks, each sensor node independently collects measurements of targets within its Field-of-View (\text{FoV}) and executes an on-board stochastic multi-target filter that produces a multi-object posterior. These posterior probabilities are then exchanged across the network and fused locally at each sensor node, allowing each sensor to have situational awareness of the larger environment. This capability is essential for multi-target tracking systems, where continuous real-time information fusion from multiple sensors enables safe and efficient autonomous operation in complex and dynamic environments.

The need for multi-sensor control emerges when sensor states, such as position and orientation, can be adjusted via control commands, i.e., \textit{movement}. An instance illustrating the utilization of multiple sensors mounted on Unmanned Air Vehicles (UAVs) for inspecting the sea surface to detect and track sea vessels is depicted in Figure~\ref{fig:overview}. In this illustration, each UAV is required to autonomously determine its subsequent actions (in terms of transitioning to a new location and orientation), relying on the hypothesized future movements of other UAVs within the network. The decision must be made in a manner that optimizes the acquisition of \textit{most informative} measurements by the interconnected UAV network, aiming to obtain a situational awareness that is as comprehensive as possible.

Multi-sensor control is a complex and critical area of research that has gained significant attention in various fields, including but not limited to autonomous vehicle networks~\cite{ref:AutonomousVehicleControl}, surveillance and intelligence applications~\cite{ref:SurveillanceNetwork}, and defense applications of multi-target tracking~\cite{ref:MTTControl, ref:GroundAreaDetection}. In all such applications, a network of sensors are used to acquire and fuse information for situational awareness. Depending on how information is communicated in the network, it may be characterized as a centralized, decentralized, or distributed sensor network~\cite{ref:NetworkArchitecture,ref:DistributedArchitecture}. 

\begin{figure}[t]
	\centering
	\includegraphics[width=\columnwidth]{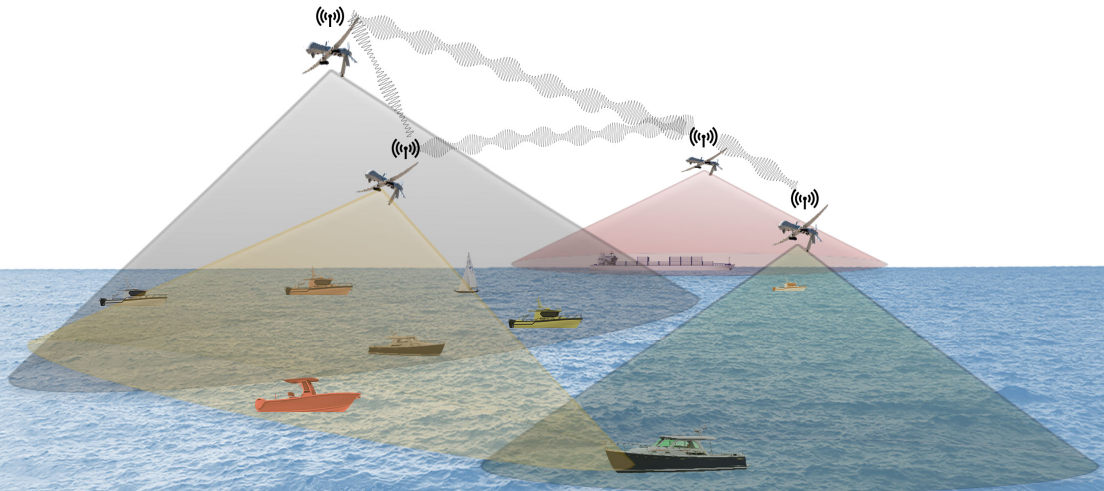}
	\caption{Multiple UAVs must cooperatively inspect part of the ocean for monitoring of marine vehicles on the sea surface.}
	\label{fig:overview}
\end{figure}

Centralized networks are made up of a central processing node that gathers data from every other node in the network, processes it, and distributes the results back to the other nodes. In decentralized networks, nodes form local clusters, with each cluster containing a central node that communicates with the other nodes in the cluster. In addition, these central nodes can communicate with other central nodes in the network, facilitating information transfer throughout the network. Distributed networks lack central processing nodes. Instead, nodes can only communicate with their nearest neighbors, usually within a limited communication range.  Distributed networks offer a pragmatic solution for implementing large-scale sensor networks, as the computational cost does not grow exponentially with the number of sensors as it does with centralized networks. In practical applications with large sensor networks, distributed communication and computing is the preferred option. Hence, we focus on distributed multi-sensor control.


Our proposed solution employs Random Finite Sets (RFS), a powerful Bayesian framework for multi-object tracking, wherein the sets of target states and observations are represented as random finite sets. Compared to other multi-target tracking frameworks such as particle filters and multiple hypothesis tracking~\cite{ref:JSCTrackingFusion}, the RFS framework has several advantages including implicitly handling dynamic numbers of targets and data association. A number of filters have been developed in this framework, including the well-known Probability Hypothesis Density (PHD)~\cite{ref:MahlerBook} and Multi-Bernoulli (MeMBer)~\cite{CBMEMBER_filter} filters. The latest development in RFS filters, called \textit{labeled RFS filters}, appends the label of each target into its single-target state and propagates target labels with their states to directly create \textit{target trajectories}~\cite{Vo2013,LMB_Vo2,reuter2015multiple}. RFS filters have found applications in various fields, including cell lineage tracking~\cite{ref:CellTracking}, Intelligent Transport Systems (ITS)\cite{ref:NidaInteraction}, and information fusion\cite{Mahler2014AdvancesIS, ref:InformationFusionRFS}.

Various sensor control methods have been proposed for RFS filters using different objective functions~\cite{ref:BrankoRistic,ref:MultiBernoulliControl,ref:VoControl, ref:HoangControl}. However, these methods have been explored mainly in single-sensor control scenarios. Recent works have introduced multi-sensor control for labeled RFS filters. For example, Jiang et al.~\cite{ref:MengMultiSensorControl} proposed a Cauchy-Schwarz Divergence (CSD)-based objective function combined with Generalized Labeled Multi-Bernoulli (GLMB) filters. Wang~\cite{ref:Xiaoying} proposed a multi-sensor control solution in which a task-driven cost (which is a mixture of the expected localization and cardinality errors) was minimized using coordinate descent, with Labeled Multi-Bernoulli (LMB) filters running on-board each sensor node. Panicker et~al.~\cite{PANICKER2020107451} developed a multi-sensor control solution that focuses on targets of interest, again with LMB filters running on each node.

The contributions of this paper are as follows:
\begin{itemize}
    \item Introducing an Adaptive Complementary Fusion (ACF) formulation for the fusion of LMB densities that prioritizes informative sensors, improving fusion accuracy.
    \item Introducing the Fully Distributed Coordinate Descent Sensor Control (FDCD-SC) algorithm, optimized for use with LMB filters across sensor nodes in a distributed network, significantly reducing computational complexity while enhancing tracking accuracy and scalability.
    \item Proposing a constrained information-theoretic objective function tailored for the FDCD-SC algorithm, yielding improved performance over existing methods.
\end{itemize}

The structure of the paper is as follows: Section~\ref{sec:Related_Work} reviews current multi-sensor control methods with a focus on RFS filter-based approaches; Section~\ref{sec:Background} provides background information on the RFS framework and formulates the problem of multi-sensor control within the general distributed multi-target tracking application; Section~\ref{sec:Solution} details the proposed ACF formulation and FDCD-SC algorithm; Section~\ref{sec:Experiment} presents experimental results, comparing our approach with state-of-the-art techniques; and Section~\ref{sec:Conclusion} concludes the paper.

\section{Related Work}
\label{sec:Related_Work}

In the RFS framework, various solutions have been developed for multi-target tracking-related problems. Some examples include track-before-detect visual tracking applications~\cite{s20030929}, sensor management in target tracking applications~\cite{PANICKER2020107451}, and information fusion~\cite{9311857}. Several information fusion approaches have been explored for LMB filters, ranging from Cooperative CSD-based fusion~\cite{ref:CauchySchwarzLMB} to Consensus-based fusion~\cite{ref:ConsensusLMB} and Complementary fusion~\cite{ref:Klupacs, ref:KlupacsFusion, ref:KlupacsComplementaryGCI}. 

The communication between sensors in a distributed sensor network has been the subject of significant research, but distributed multi-sensor control has not received as much attention. Akselrod and Kirubarajan~\cite{ref:MarkovDecisionProcesses} proposed a distributed control algorithm for multi-target tracking by a swarm of UAVs, where measurements and associated tracks from each UAV are broadcast throughout the network and fused with the local measurements at each UAV. However, each UAV chooses an action to take without considering the actions of other UAVs, which could lead to multiple nearby UAVs converging on the same target, which is counterproductive in a large distributed network. To address this issue, Fu and Yang~\cite{ref:MultitargetTrackingMobileSensor} proposed a control framework aimed at maximizing tracking accuracy while guaranteeing tracking coverage. Each sensor can be allocated to track a particular target or group of targets, and sensor positions and allocations are shared distributively throughout the network to jointly optimize for tracking accuracy while ensuring that at least one sensor is allocated to every target.

Yuan, Zhan, and Li~\cite{ref:DecentralizedQuadcopterControl} use a distinct approach by flocking quadcopter UAVs in a decentralized manner. Each UAV shares its local information with neighboring nodes, which is then used by a decentralized model predictive control flocking algorithm to form a quasi \(\alpha\)-lattice. In this application, sensors have access to information only from their neighboring sensors, limiting their ability to form a fully accurate lattice. Li et al.~\cite{ref:DistributedSparsityControl} recently proposed a fully distributed multi-sensor control approach that uses auctioned POMDPs to determine each sensor's action. While this approach is distributed and computationally efficient, it has the limitation of each possible action being taken by at most one sensor at a time.

Recently, several diverse RFS-based multi-sensor control approaches have been developed. These include a decentralized multi-sensor control algorithm combining Thompson sampling with PHD filters designed for the challenging task of targets outnumbers sensors~\cite{ref:TargetsOutnumberAgents}, a method for tracking UAV swarms using a likelihood function that integrates swarm characteristics with LMB filters to improve tracking performance~\cite{ref:TharaniControl}, and a deep reinforcement learning-based distributed multi-sensor control approach using LMB filters that learns the targets' motion models and inherently balances exploration and exploitation~\cite{ref:DeepReinforcementLearning}. Additionally, coordinate descent has recently been used for multi-sensor control optimization in RFS-based multi-object tracking applications, both in a centralized manner~\cite{ref:Xiaoying} and a distributed manner~\cite{ref:AidanICCAIS}.

Generally, two types of objective functions have appeared in the sensor control literature: the \textit{task-driven} and the \textit{information-driven}. The former type is usually defined to directly optimize a particular aspect of tracking. For instance, in some applications, the highest priority is given to the \textit{coverage} aspect. Having a network of sensors with limited fields of view, distributing the sensors for maximum coverage could be achieved by choosing the \textit{cardinality estimate} returned by the fused pseudo-posterior. The resulting sensor control solution, originally designed for one sensor, is called ``Posterior Expected Number of Targets'' (PENT)~\cite{mahler2004probabilistic}. Alternatively, estimation of tracking error could be the particular aspect of the tracking task being optimized by sensor control. In that case, the objective function would be formulated as a cost function dependent on the expected error of estimation or tracking after a sensor action is taken~\cite{ref:PEECS,ref:MultiBernoulliControl,ref:RobustMultiBernoulliControl}. Minimizing such task functions leads to selecting the action that is statistically expected to return minimum error. 

Information-driven objective functions are usually information-theoretic reward functions designed to maximize the expected \textit{information gain} from the prior density to the posterior after a sensor action occurs, typically quantified using an information divergence from the prior to posterior. Commonly used divergences in stochastic sensor control are the R{\'e}nyi divergence~\cite{ref:Divergences,ref:HoangControl} and CSD~\cite{ref:CauchySchwarzLMB,ref:MengMultiSensorControl,ref:SabitaCauchySchwarz}. Some objective functions combine information-driven and task-driven components. Li et al.~\cite{ref:DistributedSparsityControl} implement a sparsity-promoting objective function, a weighted sum of the normalized CSD and a sparsity-promoting term, that discourages multiple sensors from observing a subset of targets and ignoring the other targets. Recently, a multi-sensor control objective function using the Kullback-Leibler divergence has been proposed~\cite{ref:AidanICCAIS, ref:AidanArxiv}.

\section{Problem Statement and Background}
\label{sec:Background}

\subsection{Overall architecture}
This work looks at the situation of a distributed network of $\mathcal{S}$ monostatic sensor nodes, denoted $\mathbb{S}=\left\{1,\ldots,\mathcal{S}\right\}$. Each node $s\in\mathbb{S}$ can transceive information with a subset of the other nodes in the network, typically those within a certain distance and without occlusion. This collection of neighboring nodes of $s$ is denoted $N^{(s)}$. The aim of the sensor network is to accurately track multiple dynamic targets in the environment. Each sensor node is capable of acquiring measurements of targets within a limited FoV and is running a local RFS filter $\pi_s$ on-board to track the target states.

Figure~\ref{fig:overall_blockdiagram} provides a general overview of the entire sequence of operations executed by each sensor node, $s$, in the sensor network. For every filtering time instance $k$, the node $s$ carries out the prediction step of its multi-object filter, both on its own multi-object prior, $\pi_{k-1,s}$, and on the latest multi-object posterior densities, $\pi_{k-1,s'}$, obtained at the previous filtering time $k-1$ from all sensor nodes $s'\in \mathbb{S}$.

The primary focus of this work is the design of a solution for the \textit{multi-sensor control} block, highlighted in red. This block takes as input all the predicted multi-object densities, $\pi_{k|k-1,s'}, \, s'\in \mathbb{S}$, and then determines the \textit{optimal} control command (e.g., movement or rotation), $u_{k,s}^*$, for sensor node $s$. Following this, the system executes the control command on the sensor (e.g., moving or rotating it as required) and subsequently performs detection, yielding a measurement set $Z_{k,s}$. This set is utilized to update the predicted density, resulting in the local posterior $\pi_{k,s}$.

The obtained local posterior is then transmitted to all neighboring sensors and serves as the local multi-object prior for the next time step. To achieve situational awareness regarding the number and states of objects of interest, each node fuses its local posterior with all received posteriors. It then estimates the multi-object state from the fused posterior.

\begin{figure}
    \centering
    \includegraphics[width=\columnwidth]{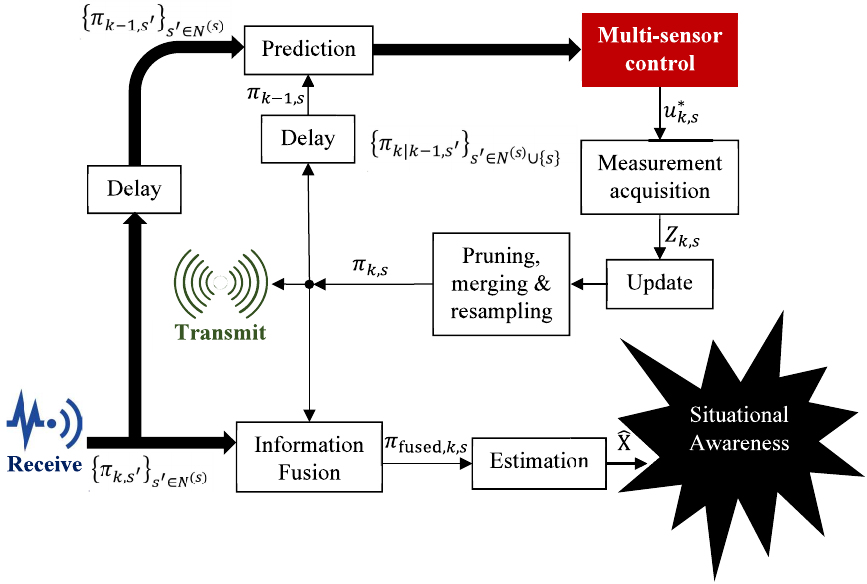}
    \caption{Overall operations running onboard each sensor node $s$, for tracking and control.}
    \label{fig:overall_blockdiagram}
\end{figure}

\subsection{System components and model assumptions} 
The prominent component of the overall architecture shown in Figure~\ref{fig:overall_blockdiagram} is the \textit{multi-object filter} running at each node. Before the filter of choice is presented, some notation and assumptions need to be clarified.

For a multi-object system at discrete filtering time step \(k\), the multi-object state $\bm{X}_k$ can be described by a labeled RFS composed of $N_k=\lvert\bm{X}_k\rvert$ single-object states,
\begin{equation}
    \bm{X}_k=\{(x_{1,k},\ell_1),\ldots,(x_{N_k,k},\ell_{N_k})\}\in\mathcal{F}(\mathbb{X}\times\mathbb{L})
\end{equation}
where each state \(x\in\mathbb{X}\) has been augmented with a label \(\ell\in\mathbb{L}\), where $\mathbb{X}$ and $\mathbb{L}$ are the single-object state space and label space respectively, and $\mathcal{F}(\mathbb{X}\times\mathbb{L})$ denotes all finite subsets of $\mathbb{X}\times\mathbb{L}$. At each \(k\), each sensor node $s$ of the network returns a set of measurements $Z_k^{(s)}=\{z_{1,k}^{(s)},\ldots,z_{m,k}^{(s)}\}\in\mathcal{F}(\mathbb{Z})$, where $\mathbb{Z}$ is the measurement space, that may include object detections and false alarms (clutter).

If a sensor is located at $(p_{x,s},p_{y,s})$, with bearing $\theta_s$, and an object is located at $(p_x,p_y)$, then the range and relative bearing of the object relative to the sensor are given by:
\begin{equation}
    \rho =\sqrt{\left(p_x-p_{x,s}\right)^2+\left(p_y-p_{y,s}\right)^2}
    \label{eq:range}
\end{equation}
\begin{equation}
    \theta =\mathrm{atan2}\left(p_x-p_{x,s},p_y-p_{y,s}\right)-\theta_s,
    \label{eq:bearing}
\end{equation}
where $\mathrm{atan2}$ means the four-quadrant inverse tangent function.

The probability of detection of an object with state $x=(\rho,\theta)$ by sensor $s$ at time step $k$ is denoted by $p_{D,k}^{(s)}(x)$ and is dependent on the time step, sensor state, and object state. The \textit{limits} of the sensor's FoV is mathematically formulated by incorporating these dependencies. An example is shown in Figure~\ref{fig:FoV_limits}. To model the probability of detection of an object decreasing when approaching the limits of the sensor's FoV, a sigmoid-based model is used. Using~\eqref{eq:range} and~\eqref{eq:bearing}:
\begin{equation}
        \setlength{\arraycolsep}{4pt}
    p_{D,k}^{(s)}(x) = \left\{
    \begin{array}{ll}
        p_{D,\rho}(\rho)\cdot p_{D,\theta}(\theta)
        & \text{if }|\theta|\leqslant \theta_{\max},\rho\leqslant\rho_{\max}\\
        0 & \text{else}
    \end{array}
    \right.
        \setlength{\arraycolsep}{6pt}
    \label{eq:pd_example}
\end{equation}
where
\begin{equation}
    p_{D,\rho}(\rho)=\frac{p_{D,\max}}{1+\mathrm{exp}\big(k_{\rho}\cdot(\rho-\rho_{\mathrm{max}})\big)}
\end{equation}
\begin{equation}
    p_{D,\theta}(\theta)=\frac{p_{D,\max}}{1+\mathrm{exp}\big(k_{\theta}\cdot(\lvert\theta\rvert-\theta_{\mathrm{max}})\big)}
\end{equation}
where $\rho_{\max}$ and $\theta_{\max}$ are the maximum range and relative bearing that will return a measurement, and $k_{\rho}$ and $k_{\theta}$ are constants that control the sharpness of the decrease.

\begin{figure}
    \centering
    \includegraphics[width=0.6\columnwidth]{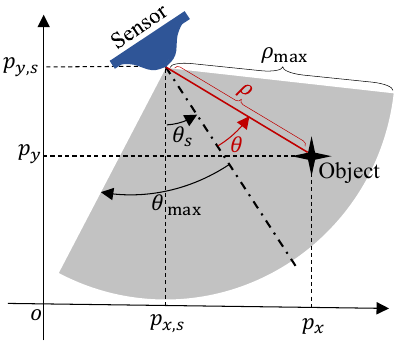}
    \caption{Diagram of a sensor with limited FoV, and the notation used for formulation of the probability of detection.}
    \label{fig:FoV_limits}
\end{figure}

Conditional on detection, the sensor $s$ is assumed to provide a measurement $z$ that is distributed according to a density given by the likelihood function \(g_{s}({\cdot|x,x_s})\), where $x$ is the detected object's state and $x_s$ denotes the sensor's state (e.g., its location $(p_{x,s},p_{y,s})$ and orientation $\theta_s$ as shown in Figure~\ref{fig:FoV_limits}). Each sensor can be controlled via a sensor control command $u^*_s\in\mathbb{U}$ where \(\mathbb{U}\) is a finite set of sensor control commands. Therefore, the sensor's state will depend on what control action is chosen and executed, hence the measurement likelihood values will be dependent on sensor control actions too. A multi-sensor control command can be constructed by appending multiple single-sensor control commands into a tuple, \(\mathfrak{u}=(u_1,\ldots,u_{\mathcal{S}})\in\mathbb{U}^{\mathcal{S}}\). 

The overarching aim of controlling a set of sensors for multi-object tracking is to substantially improve the accuracy of target state estimation. This encompasses refining both the detection of the total number of targets (cardinality) and the accurate assessment of each target's state.

\subsection{The multi-object filter}
As was emphasized previously, a multi-object filter runs onboard each sensor node. The Bayesian multi-object filter propagates the multi-object density through its prediction and update steps, and after the posterior is fused with other similar posteriors received from neighboring nodes, it extracts the multi-object estimate, which is our realization of situational awareness. The distributed multi-sensor control solution proposed in this paper can be applied with almost any choice of RFS-based multi-object filter. However, due to its proven performance and intuitive structure, we present and formulate our solution for applications where the LMB filter~\cite{reuter2015multiple} is locally running in each node. 

The LMB distribution is completely described by its components \({\bm{\pi}}=\{(r^{(\ell)},p^{(\ell)}(\cdot))\}_{\ell\in\mathbb{L}}\) where \(r^{(\ell)}\) is the probability of the existence of an object with label \(\ell\in\mathbb{L}\), and \(p^{(\ell)}(x)\) is the probability density of the object state conditional on its existence. Mathematically, for any multi-object state $\bm{X} = \{(x_1,\ell_1),\ldots,(x_n,\ell_n)\}$, the multi-object density is given by
\begin{equation}
    \bm{\pi}(\bm{X})=\Delta(\bm{X})\,\omega(\{\ell_1,\ldots,\ell_n\})\prod_{i=1}^{n} p^{(\ell_i)}(x_i)
\end{equation}
where
\begin{equation}
    \Delta(\bm{X})=\begin{cases}
        1 & \text{if }\lvert\bm{X}\rvert=\lvert\{\ell_1,\ldots,\ell_n\}\rvert\\
        0 & \text{otherwise}
    \end{cases}
\end{equation}
is the distinct label indicator and
\begin{equation}
    \omega(L)=
    \prod_{\ell\in L} r^{(\ell)}
    \prod_{\ell'\in (\mathbb{L} \setminus L)}(1-r^{(\ell')})
\end{equation}
is the probability of joint existence of all objects with labels \(\ell \in L\) and non-existence of all other labels. For computing single-object densities, a particle implementation is used in this paper, where the density of each LMB component with label \(\ell\) is approximated by \(J^{(\ell)}\) weights and particles, 
\begin{equation}
    p^{(\ell)}(x)\approx\sum_{j=1}^{J^{(\ell)}}w_j^{(\ell)}\delta(x-x_j^{(\ell)})
    \label{eq:particle_implementation}
\end{equation}
where \(\delta(\cdot)\) is the Dirac delta function. $\bm{\pi}_{k,s}$ denotes the local multi-object posterior density at sensor node $s$ of $\bm{X}_k$, the multi-object state at time $k$.

\begin{figure*}[t]
    \centering
    \includegraphics[width=\textwidth]{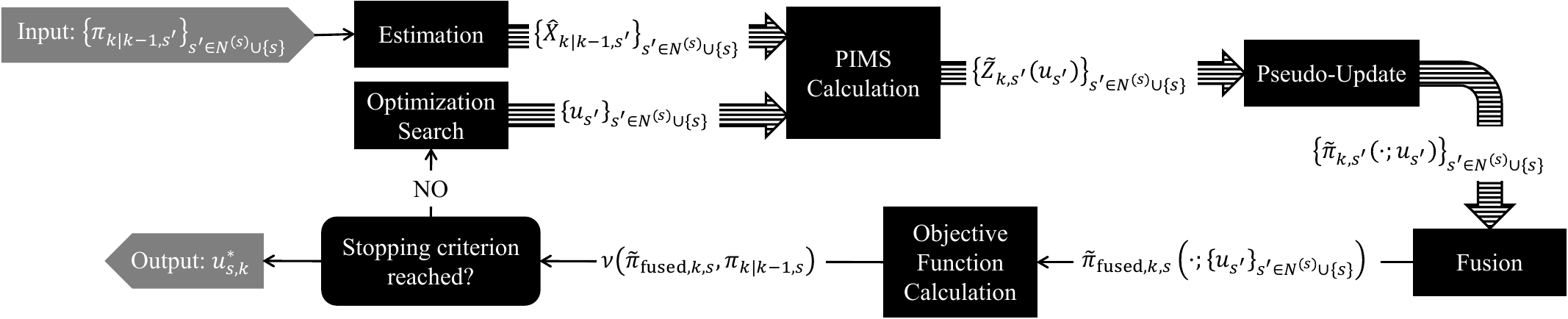}
    \caption{The proposed architecture of the contents of the multi-sensor control block executing at sensor node $s$.}
    \label{fig:proposed_architecture}
\end{figure*}

\begin{figure*}[t]
    \centering
    \begin{tabular}{ccc}
        \includegraphics[width=0.33\textwidth]{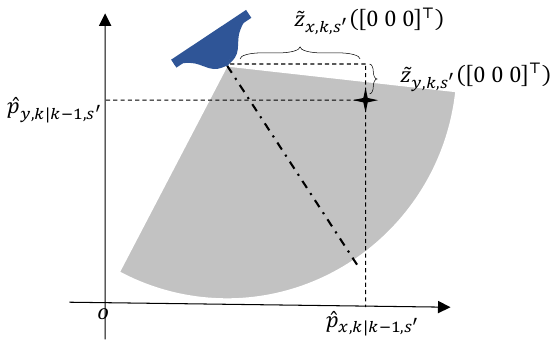}
        &
        \includegraphics[width=0.33\textwidth]{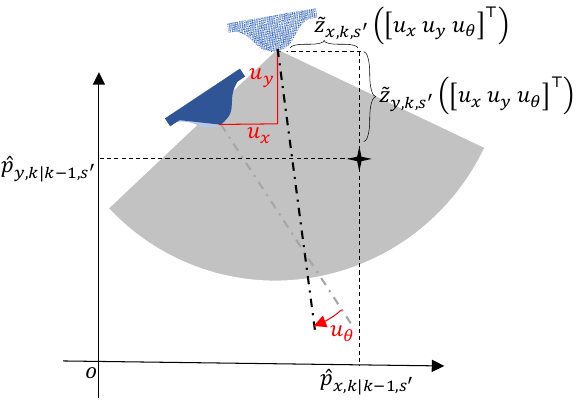}
        &
        \includegraphics[width=0.25\textwidth]{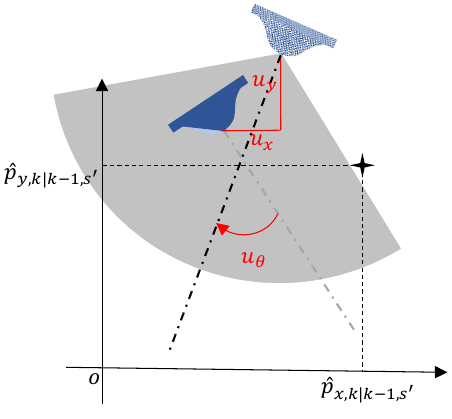}
        \\
        \footnotesize(a) & \footnotesize(b) & \footnotesize(c)
    \end{tabular}
    
    \caption{PIMS measurement returned for an object with hypothesized (a) zero control command, (b) a non-zero translation and rotation control command, and (c) the same translation as (b) but a larger rotation that leads to the object being missed.} 
    \label{fig:PIMS_examples}
\end{figure*}

\subsection{Distributed information fusion}
As shown in Figure~\ref{fig:overall_blockdiagram}, \textit{information fusion} is a critical task that needs to be completed in each sensor node. With limited FoVs, a \textit{complementary} fusion method that combines information from all the sensors is required. Klupacs et al.~\cite{ref:KlupacsFusion} recently introduced a fully distributed implementation of complementary information fusion for LMB densities that works with sensors with limited FoVs and handles label inconsistency and double counting, two significant challenges in distributed information fusion.

However, it has a critical flaw when used for multi-sensor control. When there is a large network of sensors with limited FoVs, it is likely that the following scenario will occur: if only a single sensor $s$ is observing a particular target $\ell$ and the other sensors in the network are positioned such that the target is not in their FoVs, then if $s$ no longer observes the target (either from the target moving out its FoV of, $s$ taking an action that moves the target out of its FoV, or target death), then after the update step the probability of existence at $s$, $r_{k,s}^{(\ell)}$ will decrease, however at the other sensors $s'\in\mathbb{S}\setminus \{s\}$, the probability of existence after update $r_{k,s'}^{(\ell)}$ will remain high. Due to the nature of the complementary fusion rule, this results in the fused probability of existence $\bar{r}_{k,s}^{(\ell)}$ remaining high, even when the target is not being observed. Therefore, the fused probabilities of existence will be practically identical across sensor actions, making it impossible to differentiate effective and ineffective sensor actions. To address this, a different fusion method must be used, which will be introduced in Section~\ref{subsec:ACF}.

\section{Distributed Multi-Sensor Control: The Proposed Solution}
\label{sec:Solution}

In this section, we present our proposed multi-sensor control solution for distributed multi-target tracking. Figure~\ref{fig:proposed_architecture} shows a block diagram of our proposed solution, in a form that could be the contents of the ``Multi-sensor Control'' block in Figure~\ref{fig:overall_blockdiagram}.

\subsection{Estimation}
The first step is to \textit{estimate} the number of objects and their states from the predicted densities at the local node $s$ and received from the other nodes $s'\in \mathbb{S}\setminus \{s\}$. We choose the EAP estimates. In the case of LMB densities, assuming that $\bm{\pi}_{k|k-1,s'}$ is parametrized as
$
\left\{\left(r_{k|k-1,s'}^{(\ell)}, \left\{\left(w_{j,k|k-1,s'},x_{j,k|k-1,s'}\right)\right\}_{j=1}^{J_{k-1,s'}^{(\ell)}}\right)\right\}_{\ell\in\mathbb{L}_{k-1,s'}}.
$
The EAP estimate for the number of objects (cardinality) is given by the sum of all probabilities of existence,
\begin{equation}
    \hat{|\bm{X}|}_{k|k-1,s'} = \sum_{\ell\in\mathbb{L}_{k-1,s'}} r_{k|k-1,s'}^{(\ell)}.
    \label{eq:card_EAP_estimate}
\end{equation}
The objects predicted to exist are then chosen as those with the highest probabilities of existence, up to $\hat{|X|}_{k|k-1,s'}$ objects. For each chosen label $\ell$, the EAP estimate for the state is:
\begin{equation}
    \label{eq:state_EAP_estimate}
    \hat{x}^{(\ell)}_{k|k-1,s'} = \sum_{j=1}^{J_{k-1,s'}^{(\ell)}} w_{j,k|k-1,s'}\,x_{j,k|k-1,s'}.
\end{equation}


\subsection{Calculation of the Predicted Ideal Measurement Set (PIMS) \& Pseudo-Update}
The next step is to compute hypothesized measurement sets that could ideally be returned by each sensor node, to evaluate the utility of potential multi-sensor control commands without actually performing the actions and obtaining real measurements. Such sets are \textit{ideal} in the sense that they include no measurement noise, false alarms, or miss-detections, and correspond to all objects estimated to exist. An example is shown in Figure~\ref{fig:PIMS_examples}. Consider an object that is predicted by sensor node $s'$ to exist at an estimated location of $(\hat{p}_{x,k|k-1,s'},\hat{p}_{y,k|k-1,s'})$. As shown in Figure~\ref{fig:PIMS_examples}(a), without any sensor action (zero control command), ideally the object would be detected, returning a measurement that is comprised of the horizontal and vertical distance from the object to the sensor, denoted by $\tilde{z} = (\tilde{z}_{x,k,s'},\tilde{z}_{y,k,s'})$. If a hypothesized control command in the form of a 2D translation and rotation is applied to the sensor, the measurements will change -- see Figure~\ref{fig:PIMS_examples}(b). Due to the sensor's limited field of view, it is also possible that a predicted object is entirely missed by the sensor -- see Figure~\ref{fig:PIMS_examples}(c).

Once the PIMS $\tilde{Z}_{k,s}$ is generated for each sensor node, a pseudo-update step is run. ``Pseudo'' refers to the hypothetical nature of the measurements generated with zero noise. The results are the multi-object posteriors that would be generated locally in each sensor node if the node's locally acquired measurements were the PIMS generated in node $s$ for them. The pseudo-posterior at time step $k$, at node $s$, using the PIMS generated at $s$, is denoted as $\tilde{\bm{\pi}}_{k,s}(\cdot)$.

\subsection{Adaptive Complementary Fusion}
\label{subsec:ACF}

To address the problems with complementary fusion outlined earlier, we propose an Adaptive Complementary Fusion (ACF) method that selectively includes sensors based on FoV constraints and the expected state of each target. Unlike Adaptive GCI fusion~\cite{ref:KlupacsDistributedWeightedGCI}, where a weighted sum of each sensors' contribution is used, complementary fusion is formulated using the union of each sensors' labeled random finite set, where for each label $\ell$, $Y^{(\ell)}=\bigcup_{s=1}^{\mathcal{S}}X_s^{(\ell)}$, which does not have a well-defined way to weight the contributions of each sensor~\cite{ref:ComplementaryFusion}. Instead, we propose restricting the set of sensors included in the union to just the sensors that contribute useful information. We now define the fused labeled random finite set for label $\ell$ to be the union over the set of active sensors for $\ell$,
\begin{equation}
    \mathbf{Y}_k^{(\ell)}=\displaystyle\bigcup_{s\in\mathcal{A}_k^{(\ell)}} \mathbf{X}_{k,s}^{(\ell)}.
\end{equation}

The set of active sensors for label $\ell$ is defined as
\begin{equation}
    \mathcal{A}_k^{(\ell)}=\left\{s\in\mathbb{S}:1_{\text{FoV}}^{(\ell)}(s,k)=1\right\},
\end{equation}
where
\begin{equation}
1_{\text{FoV}}^{(\ell)}(s,k)=\text{max}\left(1_{\text{update}}^{(\ell)}(s,k),1_{\text{prediction}}^{(\ell)}(s,k)\right)
\end{equation}
and
\begin{equation}
    1_{\text{update}}^{(\ell)}\left(s,k\right)=\begin{cases}
        1 & \text{if }p_{D,k}^{(s)}(\hat{x}_{k,s}^{(\ell)})>p_{D,\text{th}}\\
        0 & \text{otherwise}
    \end{cases}
\end{equation}
\begin{equation}
    1_{\text{prediction}}^{(\ell)}\left(s,k\right)=\begin{cases}
        1 & \text{if }p_{D,k}^{(s)}(\hat{x}_{k|k-1,s}^{(\ell)})>p_{D,\text{th}}\\
        0 & \text{otherwise},
    \end{cases}
\end{equation}
with $p_{D,\text{th}}$ as the threshold probability of detection. By substituting $\mathcal{A}_k^{(\ell)}$ for $\mathbb{S}$ in the complementary fusion rules~\cite{ref:ComplementaryFusion}, the Adaptive Complementary Fusion rules for LMB densities are:
\begin{equation}
    \bar{r}_{k}^{(\ell)}=\displaystyle\sum_{s\in\mathcal{A}_k^{(\ell)}}\frac{r_{k,s}^{(\ell)}}{1-r_{k,s}^{(\ell)}}\Bigg/\left[1+\displaystyle\sum_{s\in\mathcal{A}_k^{(\ell)}}\frac{r_{k,s}^{(\ell)}}{1-r_{k,s}^{(\ell)}}\right]
\end{equation}
\begin{equation}
    \bar{p}_{k}^{(\ell)}(\cdot)=\displaystyle\sum_{s\in\mathcal{A}_k^{(\ell)}}\frac{r_{k,s}^{(\ell)}p_{k,s}^{(\ell)}(\cdot)}{1-r_{k,s}^{(\ell)}}\Bigg/\displaystyle\sum_{s\in\mathcal{A}_k^{(\ell)}}\frac{r_{k,s}^{(\ell)}}{1-r_{k,s}^{(\ell)}}.
\end{equation}

This means that if either the updated or predicted target position is inside the FoV of sensor $s$, then $s$ will be included in that target's active sensor set. If the updated target position is inside a sensor's FoV, then that sensor has measured the target and will clearly contribute useful information for fusion. If the updated target position is not inside a sensor's FoV, either due to a miss-detection, target death, or the target has moved unexpectedly outside of the sensor's FoV, but the predicted target position is inside the sensor's FoV, then that is also useful information for fusion. If neither case is true, then the sensor was not expected to be able to measure the target and did not measure the target; therefore, it did not collect any useful information for fusion and is excluded from $\mathcal{A}_k^{(\ell)}$.

If $\lvert\mathcal{A}_k^{(\ell)}\rvert>1$, then all elements of $\mathcal{A}_k^{(\ell)}$ are fused using the ACF rules. If $\lvert\mathcal{A}_k^{(\ell)}\rvert=1$, then just a single sensor is active; therefore, fusion does not need to take place, and that sensor's density is used. The last case, $\mathcal{A}_k^{(\ell)}=\emptyset$, requires special handling. When no sensor in the network has measured target $\ell$ or was expected to measure $\ell$, discarding all Bernoulli components for $\ell$ would result in that target being dropped, which is not necessarily desired. Therefore, two different approaches are used, depending on the type of densities being fused. \textbf{Pseudo-update fusion}: If fusing pseudo-updates, then targets that are not observed or expected to be observed by any sensor should be discarded, as they do not provide any information for sensor control. Therefore, $\mathcal{A}_k^{(\ell)}$ is left empty and label $\ell$ is omitted from the fused density. \textbf{Update fusion}: If fusing updates, then targets should be retained even if they are not observed; for example, if a sensor may perform an action that results in a target no longer being observed, then that target should still be tracked. Therefore, all sensors are assumed to contribute equally and $\mathcal{A}_k^{(\ell)}=\mathbb{S}$.

Note that restricting the set of sensors that contribute to the fusion of target $\ell$ to the set $\mathcal{A}_k^{(\ell)}$ addresses the issue of delays in dropping targets after death when using complementary fusion. This issue is caused by the fused probability of existence of a target being high if just one of the sensors' probabilities of existence is high, e.g., when a sensor does not observe a target. Therefore, even if the death of a target $\ell$ is directly observed by some sensors in the network, their respective low probabilities of existence are 'overruled' by the sensors that were not observing that area and thus have high probabilities of existence. The proposed ACF method resolves this, quickly reducing a target's probability of existence after death.

The proposed ACF rules are combined with the distributed consensus approach, extended label space, and label merging from Klupacs' distributed fusion approach~\cite{ref:KlupacsFusion} in this work.
	
\subsection{Proposed objective function}
	
The pseudo-posterior densities $\{\tilde{\bm{\pi}}_{k,s}(\cdot):s\in\mathbb{S}\}$ and their fused density $\tilde{\bm{\pi}}_{\mathrm{fused},k,s}(\cdot)$ are clearly dependent on the pseudo-measurements which are themselves dependent on a hypothetical multi-sensor control command, $\tilde{\mathfrak{u}} = (\tilde{u}_1,\ldots,\tilde{u}_{\mathcal{S}})\in\mathbb{U}^{\mathcal{S}}$. The fused pseudo-posterior is input to an objective function, which is at the core of optimizing the sensing performance through sensor control. It returns a scalar value and should be formulated in such a way that its optimal value is associated with the \textit{best statistically expected} sensing performance.

In this work, we propose an information-theoretic objective function subject to two constraints based on practicability: avoiding collisions with the targets and avoiding collisions with the other sensors.

\subsubsection{Kullback-Leibler Divergence}

The Kullback-Leibler Divergence (KLD) is a statistical measure of how different a model distribution $Q$ is from a reference distribution $P$. For continuous distributions with probability density functions $p$ and $q$,

\begin{equation}
    D_{\text{KL}}(P\,||\,Q)=\int_{-\infty}^{\infty} p(x)\,\mathrm{log}\left(\frac{p(x)}{q(x)}
    \right)dx.
    \label{eq:KLD}
\end{equation}

Note that the `$||$’ symbol in~\eqref{eq:KLD} does not represent norm; it separates the two probability distributions $P$ and $Q$, indicating ``the divergence of Q from P’’. If the two distributions are instead RFS densities $\bm{\pi}_1$ and $\bm{\pi}_2$, then

\begin{equation}
        D_{\text{KL}}
        \Big(
        \bm{\pi}_1\,||\,\bm{\pi}_2
        \Big)
        =	
        \int_{\mathbb{X}} \,\bm{\pi}_{1}(X) 
        \,\log
        \left(
        \dfrac{\bm{\pi}_{1}(X)}{\bm{\pi}_{2}(X)}
        \right)\, \delta X.
    \label{eq:KLD_reward}
\end{equation}
where $\delta X$ denotes the set integral~\cite{ref:MahlerBook}. In the case that both densities are LMB densities, denoting them by:
$$
\bm{\pi}_{1} = 
    \left\{
        \left(r^{(\ell)}_1,p^{(\ell)}_1\right)
    \right\}_{\ell\in\mathbb{L}_1}
$$
$$
\bm{\pi}_2 = 
    \left\{
        \left(r^{(\ell)}_2,p^{(\ell)}_2\right)
    \right\}_{\ell\in\mathbb{L}_2}.
$$
Assuming label consistency, the integral has a closed form given by~\cite{ref:LMBParticle}:
\begin{equation}
\begin{split}
    &D_{\text{KL}}\Big(\bm{\pi}_1 \,||\, \bm{\pi}_2\Big) =
    \sum_{\ell\in\mathbb{L}_1\cup\mathbb{L}_2}
    \Bigg[
        r^{(\ell)}_1 \log\left(\frac{r^{(\ell)}_1}{r^{(\ell)}_2}\right)
        + \\ 
        &\,\,\quad\left[1-r^{(\ell)}_1\right]\log\left(\frac{1-r^{(\ell)}_1}{1-r^{(\ell)}_2}\right) + r^{(\ell)}_1 
        D_{\text{KL}}\left(p^{(\ell)}_1\,||\,p^{(\ell)}_2\right)
    \Bigg],
    \label{eq:LMB_reward}
\end{split}
\end{equation}
where the final term is the KLD between two spatial distributions,
\begin{equation}
    D_{\text{KL}}\left(p^{(\ell)}_1\,||\,p^{(\ell)}_2\right)=\int p_1^{(\ell)}(x)\log\left(\frac{p_1^{(\ell)}(x)}{p^{(\ell)}_2(x)}\right)dx.
    \label{eq:Bernoulli_KLD}
\end{equation}

For computational tractability and to focus on tracking a maximal number of targets, we omit the spatial component of~\eqref{eq:LMB_reward} and just consider the divergence in existence probabilities:
\begin{multline}
    D_{\text{KL}}^{\text{exist}}\Big(\bm{\pi}_1(\cdot) \,||\, \bm{\pi}_2(\cdot)\Big) =
    \sum_{\ell\in\mathbb{L}_1\cup\mathbb{L}_2}
    \Bigg[
        r^{(\ell)}_1 \log\left(\frac{r^{(\ell)}_1}{r^{(\ell)}_2}\right)
        + \\ 
        \left[1-r^{(\ell)}_1\right]\log\left(\frac{1-r^{(\ell)}_1}{1-r^{(\ell)}_2}\right)
    \Bigg].
    \label{eq:LMB_reward_exist}
\end{multline}

This formulation prioritizes actions that maximize information gain regarding target presence while maintaining computational efficiency for real-time multi-sensor control.


There are two important cases to consider. Firstly, when $\ell\in\mathbb{L}_2\setminus\mathbb{L}_1$, then $r_1^{(\ell)}=0$ the contribution of $\ell$ to~\eqref{eq:LMB_reward_exist} simplifies to:
\begin{equation}
    -\log\left(r_2^{(\ell)}\right).
    \label{eq:dropped_target}
\end{equation}

Secondly, when $\ell\in\mathbb{L}_1\setminus\mathbb{L}_2$, then $r_2^{(\ell)}=0$, which results in the term $r_1^{(\ell)}\log\left(r_1^{(\ell)}/\,r_2^{(\ell)}\right)$ in~\eqref{eq:LMB_reward_exist} becoming infinitely large. This is a reflection of the information gain when a new target is observed in $\bm{\pi}_1$ that was not observed in $\bm{\pi}_2$. However, in order to distinguish between multiple different actions that result in new target observations, we approximate the information gain by setting $r_2^{(\ell)}=\varepsilon$, where $\varepsilon$ is a small constant, resulting in the contribution of $\ell\in\mathbb{L}_1\setminus\mathbb{L}_2$ to~\eqref{eq:LMB_reward_exist} approximates to
\begin{equation}
        r^{(\ell)}_1 \log\left(\frac{r^{(\ell)}_1}{\varepsilon}\right)
        + \left[1-r^{(\ell)}_1\right]\log\left(\frac{1-r^{(\ell)}_1}{1-\varepsilon}\right)
        \label{eq:epsilon}
\end{equation}
which is a finite positive value, allowing actions that observe a higher number of new targets to have a higher value.

\subsubsection{Constraints}

In practical multi-object tracking scenarios, it is critical that sensors avoid collisions with both the targets that they are tracking and the other sensors in the network. We propose the following constraints:

Firstly, we add a constraint on actions that would move the sensor too close to the targets. The void probability functional (or void probability) of a point process is the probability that a given region does not contain any points~\cite{ref:VoidProbability}. In the multi-object tracking problem, it is the probability that a given region does not contain any targets. This can be used to prioritize sensor actions that do not move the sensor in such a way that targets enter an exclusion area around the sensor. This exclusion area $\epsilon(\tilde{u}_s)$, a circular region centered on the sensor with radius $\rho_\epsilon$, depends on the hypothetical action $\tilde{u}_s$ taken by the sensor. The void probability of a SMC-LMB density~\cite{ref:ConstrainedCSDControl} is defined as:
\begin{equation}
    \psi_{\tilde{\bm{\pi}}}\left(\epsilon(\tilde{u}_s)\right)=\displaystyle\prod_{\ell\in\mathbb{L}}\left(1-\tilde{r}^{(\ell)}\sum_{j=1}^{J^{(\ell)}}1_{\epsilon(\tilde{u}_s)}(\tilde{x}_{j}^{(\ell)})\tilde{w}_{j}^{(\ell)}\right)
    \label{eq:VoidProbability}
\end{equation}
where
\begin{equation}
    1_{\epsilon(\tilde{u}_s)}(\tilde{x}_{j}^{(\ell)})=\left\{
    \begin{array}{ll}
        1
        & \text{if }\tilde{x}_{j}^{(\ell)}\text{ is in region }\epsilon(\tilde{u}_s)
        \\
        0 & \text{else}.
    \end{array}
    \right.
\label{eq:Indicator}
\end{equation}

Because $\sum_{j=1}^{J_k^{(\ell)}}w_{j}^{(\ell)}=1$, for a particular label, if all the particles for a particular label are inside the exclusion region, then the term in~\eqref{eq:VoidProbability} becomes $(1-\tilde{r}^{(\ell)})$. If all the particles are outside the exclusion region, it becomes $1$.

For a network of sensors with an associated hypothetical multi-sensor control command $\tilde{\mathfrak{u}}$, we define the multi-sensor sensor-target collision-avoidance term $\psi\left(\tilde{\mathfrak{u}}\right)$ as:
\begin{equation}
    \psi_{\tilde{\bm{\pi}}}\left(\tilde{\mathfrak{u}}\right)=\underset{\tilde{u}_s\in\tilde{\mathfrak{u}}}{\text{max}}\,\psi_{\tilde{\bm{\pi}}}\left(\epsilon\left(\tilde{u}_s\right)\right).
    \label{eq:SensorTargetCollision}
\end{equation}

To avoid collision with other sensors, we also define a sensor-sensor collision-avoidance term $\eta(\tilde{\mathfrak{u}})$ that is equal to the minimum distance between sensors in the network after hypothetical actions $\tilde{\mathfrak{u}}$ are taken:

\begin{equation}
    \eta(\tilde{\mathfrak{u}})=\min_{s,s'\in \mathbb{S},s\neq s'}d(s,s';\tilde{\mathfrak{u}}),
    \label{eq:collision}
\end{equation}

where

\begin{multline}
    d(s,s';\tilde{\mathfrak{u}})=\\ \sqrt{\Big(p_{x,s}(\tilde{u}_s)-p_{x,s'}(\tilde{u}_{s'})\Big)^2+\Big(p_{y,s}(\tilde{u}_s)-p_{y,s'}(\tilde{u}_{s'})\Big)^2}
\end{multline}
is the distance between sensors $s$ and $s'$ after $\tilde{\mathfrak{u}}$.

\subsubsection{Proposed Objective Function}

The multi-sensor control objective function is defined as the reward associated with choosing a hypothetical multi-sensor control command $\tilde{\mathfrak{u}}$, quantified by computing the KLD from a pseudo-posterior density that depends on that hypothetical multi-sensor control command, $\bm{\pi}_1(\cdot;\tilde{\mathfrak{u}})=\tilde{\bm{\pi}}_{k,s}(\cdot;\tilde{\mathfrak{u}})$, to the predicted density, $\bm{\pi}_2(\cdot)=\bm{\pi}_{k|k-1,s}(\cdot)$. However, as shown in~\eqref{eq:dropped_target}, when a target $\ell$ is found in $\bm{\pi}_{k|k-1,s}(\cdot)$ but not in $\tilde{\bm{\pi}}_{k,s}(\cdot;\tilde{\mathfrak{u}})$, this adds a finite positive component to the objective score, which rewards actions that result in targets being dropped. To counteract this, we add a penalty when a label is present in the prediction but not in the pseudo-posterior to instead penalize actions that result in targets being dropped, which is formulated as:
\begin{equation}
    \phi\left(\bm{\pi}_1,\bm{\pi}_2\right)=-\lambda\displaystyle\sum_{\ell\in\mathbb{L}_2\setminus\mathbb{L}_1}\mathrm{log}\left(r_2^{(\ell)}\right)
    \label{eq:penalty}
\end{equation}
where $\lambda\geq1$ is a scaling factor that ensures that for a target $\ell\in\mathbb{L}_2\setminus\mathbb{L}_1$, the sum of its component in the KLD, given by~\eqref{eq:LMB_reward_exist}, and its component in~\eqref{eq:penalty} is non-positive.

The final multi-sensor control objective function is defined as the sum of the KLD between the pseudo-posterior density and prediction density~\eqref{eq:LMB_reward_exist} and the penalty function~\eqref{eq:penalty},

\begin{equation}
\label{eq:our_objective_fun}
    \nu\left(\bm{\pi}_1(\cdot;\tilde{\mathfrak{u}}),\bm{\pi}_2\right)= \, D_{\text{KL}}^{\text{exist}}\big(\bm{\pi}_1(\cdot;\tilde{\mathfrak{u}}) , \bm{\pi}_2\big)-\phi\left(\bm{\pi}_1\left(\cdot;\tilde{\mathfrak{u}}\right),\bm{\pi}_2\right),
\end{equation}

The optimization problem is to find the multi-sensor control command $\mathfrak{u}^*$ that maximizes~\eqref{eq:our_objective_fun} constrained by~\eqref{eq:SensorTargetCollision} and~\eqref{eq:collision}:

\begin{equation}
\label{eq:final_objective_fun}
\begin{aligned}
    \tilde{\mathfrak{u}}=&\,\underset{\mathfrak{u}\in\mathbb{U}^{\mathcal{S}}}{\text{arg\,max}}\: \nu\left(\bm{\pi}_1(\cdot;\mathfrak{u}),\bm{\pi}_2\right) \\
    \text{subject to} \quad
    & \psi_{\bm{\pi}_1}(\mathfrak{u})<\psi_{\text{th}} \\
    & \eta(\mathfrak{u}) > \eta_{\text{th}}.
\end{aligned}
\end{equation}

\subsection{Optimization and stopping criterion}
	
A different choice of the multi-sensor control command, $\mathfrak{u}$, would generate a different set of fused pseudo-updated probabilities of existence, and therefore a different value for the reward function~\eqref{eq:final_objective_fun}. The ultimate aim is to find the optimal multi-sensor control command at each sensor node that would return \textit{maximum} reward. Denoting the selected optimal command by $\mathfrak{u}^* = \{u_{s}^*:s\in \mathbb{S}\}$, each node $s$ should execute the single-sensor control command $u_s^*$ to move the sensor before actual measurement acquisition is performed.

The process of determining the solution for the multi-sensor control problem is a combinatorial optimization search, requiring an iterative optimization algorithm to avoid a high computational cost. The recently proposed distributed coordinate descent method~\cite{ref:AidanICCAIS} only considers a sensor's neighborhood when calculating the multi-sensor control command, rather than the entire sensor network, significantly reducing computational cost through reducing dimensionality and parallel optimization at each sensor node. While this method works well in simple scenarios with significantly overlapping FoVs, it is not a truly distributed implementation and is thus susceptible to selecting globally suboptimal sensor actions.

We propose a fully distributed multi-sensor control algorithm, FDCD-SC. In FDCD-SC, each sensor node will iteratively update its own selected action, taking into account the actions that every other sensor has selected. In a distributed network, two sensors cannot know each other's selected action in real-time. But they can use the flooding method of inter-node communication~\cite{ref:flooding2017, ref:flooding2020} to \textit{flood-in} the information related to the action most recently selected by others, and \textit{flood-out} their own selection to others for future use in their own asynchronous time. 

After an initialization step that occurs simultaneously across all sensors, a multi-agent coordinate descent optimization process occurs. Firstly, as the multi-agent coordinate descent process is iterative, let us denote the iteration counter for this process by $t$, in contrast to filtering iterations denoted by $k$. The optimal single-sensor control command calculated at sensor $s$ in iteration $t$ is denoted as $u_s^{(t)}$. Initially, each sensor $s\in\mathbb{S}$ will perform a single-sensor control command optimization,
\begin{equation}
    u_s^{(0)} = \underset{\tilde{u}\in\mathbb{U}}{\text{arg\,max}} \:
    \nu\left(\tilde{\bm{\pi}}_{k,s}(\cdot;\tilde{u}), \bm{\pi}_{k|k-1,s}\right),
    \label{eq:initial_update}
\end{equation}
that finds the best hypothetical single-sensor control command for sensor $s$, $u_s^{(0)}$, that maximizes the KLD between the local pseudo-posterior $\tilde{\bm{\pi}}_{k,s}$ given $u_s^{(0)}$ and the local prior density $\bm{\pi}_{k|k-1,s}$. Each sensor then broadcasts the selected single-sensor control command $u_s^{(0)}$ and it's associated pseudo-posterior $\tilde{\bm{\pi}}_{k,s}(\cdot;u_s^{(0)})$ to the other sensors in the network.

Then, in subsequent iterations $t>0$, the information received from the other sensors in the network via distributed flooding communication can be utilized. In each iteration, for each sensor $s\in\mathbb{S}$ sequentially, a multi-sensor control command optimization is performed to calculate the best action for sensor $s$, given the other sensors in the network $s'\in\mathbb{S}\setminus\{s\}$ have fixed control commands $\{u_{s'}^{(t-1)}\},s'\in\mathbb{S}\setminus \{s\}$. Therefore, in each iteration $t>0$, each sensor sequentially performs the multi-sensor control command optimization,
\begin{equation}
    u_s^{(t)} = \underset{\tilde{u}\in\mathbb{U}}{\text{arg\,max}} \:
    \nu\left(\tilde{\bm{\pi}}_{\text{fused},k,s}\left(\cdot;\tilde{\mathfrak{u}}_s^{(t)}\right), \bm{\pi}_{k|k-1,s}\right),
    \label{eq:flooding_update}
\end{equation}
where
\begin{equation}
    \tilde{\mathfrak{u}}_s^{(t)}=\left[u_{1:(s-1)}^{(t)},\tilde{u},u_{(s+1):\mathcal{S}}^{(t-1)}\right]
    \label{eq:multi_sensor_definition}
\end{equation}
and $u^{(t)}_{1:(s-1)}$ denotes the $(s-1)$-tuple of the most recent decisions made by sensor nodes $1:(s-1)$, and a similar definition applies to $u^{(t-1)}_{(s+1):\mathcal{S}}$, i.e.,
\begin{equation}
    \begin{array}{rcl}
        u^{(t)}_{1:(s-1)} & \triangleq & \bigg(u^{(t)}_1,\ldots,u^{(t)}_{s-1}\bigg) \\
        u^{(t-1)}_{(s+1):\mathcal{S}} & \triangleq & \bigg(u^{(t-1)}_{s+1},\ldots,u^{(t-1)}_{\mathcal{S}}\bigg),
    \end{array}
\end{equation}
noting that the sensors $1$ to $s-1$ have already calculated and broadcast their optimal control commands for iteration $t$. Note that ACF ensures label consistency between local predictions and local/fused (pseudo-)posteriors.

Equation~\eqref{eq:flooding_update} finds the best hypothetical single-sensor control command for sensor $s$, $\tilde{u}_s$, that maximizes the KLD between the fused pseudo-posterior $\tilde{\bm{\pi}}_{\mathrm{fused},k,s}$ given the multi-sensor control command $\tilde{\mathfrak{u}}_s^{(t)}$ as defined in~\eqref{eq:multi_sensor_definition} and the local prior density $\bm{\pi}_{k|k-1,s}$. An important note to make is that solving the optimization problems in~\eqref{eq:initial_update} and~\eqref{eq:flooding_update} by exhaustive search is not computationally prohibitive. The main reason is that it involves computing the reward for just the finite number of possible single-sensor commands.


After the optimal control command for sensor node $s$ at iteration $t$ has been calculated, sensor $s$ has a new optimal multi-sensor control command,
\begin{equation}
    \mathfrak{u}_s^{(t)}=\left[u_{1:(s-1)}^{(t)},u_s^{(t)},u_{(s+1):\mathcal{S}}^{(t-1)}\right]
    \label{eq:local_optimal_multi_sensor_definition}
\end{equation}
and then broadcasts its locally decided optimal single-sensor control command $u_s^{(t)}$ and local pseudo-posterior $\tilde{\bm{\pi}}_{k,s}(\cdot;u_s^{(t)})$ throughout the network. Therefore, at each iteration $t$, each sensor $s$ will have received all of the most recently calculated optimal decisions and pseudo-posteriors from all other sensors $s' \in \mathbb{S}\setminus\{s\}$. This flooding process typically requires multiple communication iterations to ensure complete information dissemination throughout the network, particularly in large or sparsely connected sensor deployments. For LMB densities, the communication cost per sensor is known from~\cite{ref:LMBCommunication} to be $4\left(1+\left(4+10\lvert\mathbb{L}_{k,s}^{(t)}\rvert\right)\right)$ bytes, where $\lvert\mathbb{L}_{k,s}^{(t)}\rvert$ represents the cardinality of the label set at time step $k$ and iteration $t$ for sensor $s$. As the selection action is simply an integer, the total communication cost for transmitting both the selected action and the parameters defining the pseudo-posterior density is
\begin{equation}
    C_s^{(t)}=4\left(1+\left(4+10\lvert\mathbb{L}_{k,s}^{(t)}\rvert\right)\right)+1.
\end{equation}

The exact number of flooding iterations depends on the network topology and connectivity, but eventually converges to a state where all sensors have access to the updated single-sensor control command, i.e., each sensor will have access to $u_s^{(t)}$. The proposed method has a well-defined stopping criterion, where convergence of a sensor's multi-sensor control command~\eqref{eq:local_optimal_multi_sensor_definition} is mathematically guaranteed when it occurs, detailed in the following proof. 


\begin{theorem}
    In the FDCD-SC algorithm, at any sensor $s\in\mathbb{S}$, the optimal multi-sensor control command $\mathfrak{u}_s^{(t)}$ must eventually enter a cycle.
\end{theorem}

\begin{proof}
    The optimal multi-sensor control command $\mathfrak{u}_s^{(t)}$ is composed of optimal single-sensor control commands $u_i^{(t)}$ from sensors $\{i\in\mathbb{S}:i\leq s\}$ and $u_j^{(t-1)}$ from sensors $\{j\in\mathbb{S}:j>s\}$, as shown in~\eqref{eq:local_optimal_multi_sensor_definition}.

    There are a finite number of possible single-sensor control commands, with a state space of $\lvert\mathbb{U}\rvert$, therefore there are a finite number of possible multi-sensor control commands $\mathfrak{u}_s^{(t)}$, with a state space of $\lvert\mathbb{U}\rvert^{\mathcal{S}}$.
    
    In each iteration $t>0$, at each sensor in the network $s'\in\mathbb{S}$, the optimal single-sensor control command is calculated according to~\eqref{eq:flooding_update}, which is a deterministic process. Therefore, the update from $\mathfrak{u}_s^{(t)}$ to $\mathfrak{u}_s^{(t+1)}$ as in~\eqref{eq:local_optimal_multi_sensor_definition} is also deterministic.

    By the Pigeonhole Principle, there must exist iterations $t'$ and $t$ with $1<t'<t\leq\lvert\mathbb{U}\rvert^{\mathcal{S}}+1$ such that $\mathfrak{u}_s^{(t')}=\mathfrak{u}_s^{(t)}$.

    Due to the deterministic update rule, this implies that $\mathfrak{u}_s^{(t'+l)}=\mathfrak{u}_s^{(t+l)}$ for all $l\geq 0$, establishing a cycle of a maximum length of $t_{\text{cycle}}=t-t'$.

    Therefore, the FDCD-SC algorithm must eventually enter into a cycle, completing the proof.
\end{proof}

If the cycle has a length of 1, then the algorithm has converged to a single optimal solution. If the cycle has a length greater than 1, then the algorithm has converged to a limit cycle between multiple solutions. In either case, the stopping criterion is set to be whenever a cycle is first detected:
\begin{equation}
    \text{Stop at iteration}~t~\text{if}~\exists\, t' \in \{1,2,\ldots,t-1\} : \mathfrak{u}_s^{(t)} = \mathfrak{u}_s^{(t')}
    \label{eq:stopping_criterion}
\end{equation}
The sensor where the stopping criterion is met is denoted $s^*$.

In practice, once the convergence required by the stopping criterion~\eqref{eq:stopping_criterion} is met at iteration $t_{\text{end}}$ with a cycle starting at $t_{\text{start}}$, the final optimal multi-sensor control command needs to be selected. There are multiple possible strategies that can be used to select which of the multi-sensor control commands proposed at each coordinate descent iteration, i.e. $\left\{\mathfrak{u}_{s^*}^{(t)}:t\in[t_{\text{start}},t_{\text{end}}]\right\}$, will be executed. In this work, at each coordinate descent iteration $t$, each sensor $s$ will calculate and store the objective score given the newly calculated multi-sensor control command $\mathfrak{u}_s^{(t)}$ from~\eqref{eq:local_optimal_multi_sensor_definition},
\begin{equation}
    \bar{\nu}_{s}^{(t)}=\nu\left(\tilde{\bm{\pi}}_{\text{fused},k,s}\left(\mathfrak{u}_s^{(t)}\right),\bm{\pi}_{k|k-1,s}\right).
\end{equation}.
When sensor $s^*$ reaches the stopping criterion, it will find the coordinate descent iteration within the cycle that has the highest associated objective score, i.e.,
\begin{equation}
    t^*=\underset{t\in[t_{\text{start}},t_{\text{end}}]}{\text{arg\,max}}\bar{\nu}_{s^*}^{(t)}.
\end{equation}

Finally, the optimal multi-sensor control command will be set as $\mathfrak{u}_k^*=\mathfrak{u}_{s^*}^{(t^*)}$, that sensor $s^*$ will \textit{flood-out} to all other sensors in the network, and each sensor will execute their optimal single-sensor control command $u_{k,s}^*$.

\begin{figure}[t]
   \centering
   \includegraphics[width=0.8\columnwidth]{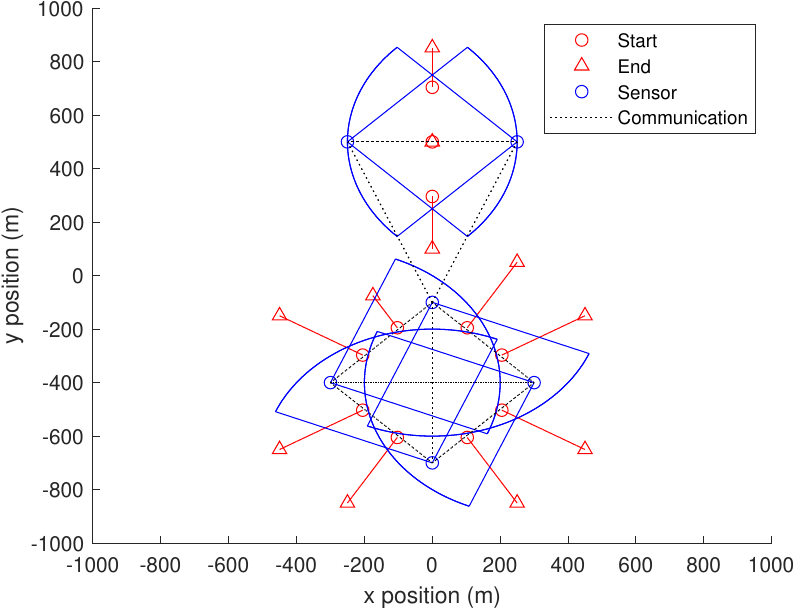}
   \caption{Overview of the 6 sensor scenario.}
   \label{fig:6sensor_scenario}
\end{figure}

\section{Numerical Experiments}
\label{sec:Experiment}
\subsection{Scenarios}

\begin{figure}[t]
   \centering
   \includegraphics[width=0.9\columnwidth]{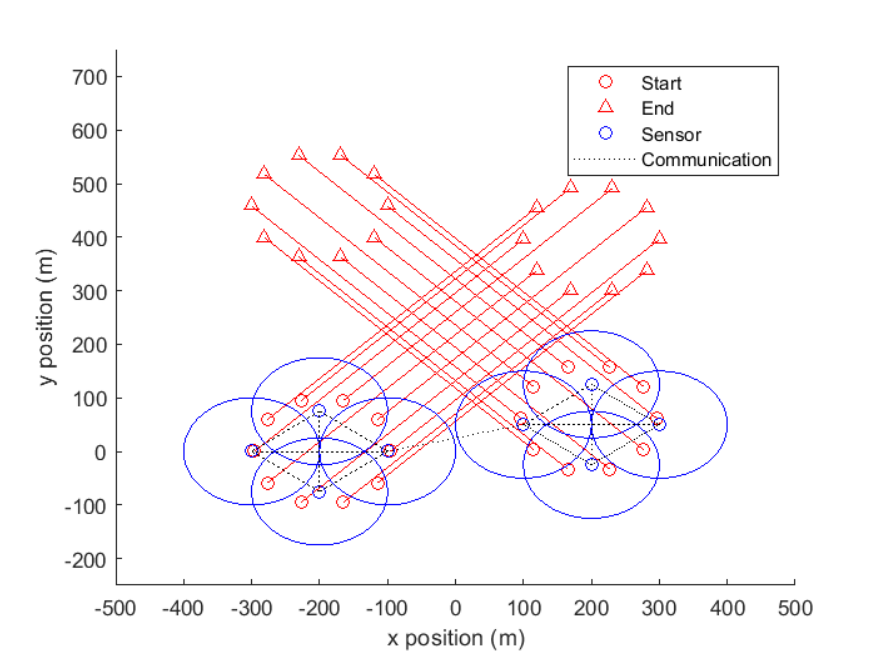}
   \caption{Overview of the 8 sensor scenario.}
   \label{fig:8sensor_scenario}
\end{figure}

Multiple challenging multi-object tracking scenario were developed to evaluate the performance of the FDCD-SC method. Figure~\ref{fig:6sensor_scenario} illustrates the first scenario, with 6 sensors tracking 11 targets for a total of 50 time steps with a period of 1 second. The targets are distributed within a $1000\,\text{m}\times 2000$\,m area and follow linear trajectories. All of the targets are present at the first time step and are observed by at least one sensor; however, as the scenario progresses, the targets will move away from the center of the scenario, following a linear motion model, and therefore out of the initial sensor FoVs. This dynamic nature of the multi-target tracking problem is handled by the complementary information fusion between sensors, so that targets aren't dropped when they move out of one sensor's FoV but remain in other sensors' FoVs, and they aren't double-counted when multiple sensors observe the same target. Additionally, two of the targets will die prematurely.

The detection profile of each sensor is modeled by~\eqref{eq:pd_example}, in the form depicted in Figure~\ref{fig:FoV_limits}, with $\theta_\max = 45^\circ$ (90$^\circ$ field of view), a maximum detection range of $\rho_\max = 500$\,m, and constants $p_{D,\mathrm{max}}=0.99$, $k_{\rho}=0.5$, and $k_{\theta}=20$. The maximum communication range between sensors is simulated to be 800\,m. Each sensor has 3 possible actions: remain stationary, rotate $22.5^\circ$ clockwise, and rotate $22.5^\circ$ anticlockwise.

Figure~\ref{fig:8sensor_scenario} illustrates the second, more challenging, scenario, with 8 sensors tracking 20 targets in a $800\,\text{m}\times 800$\,m area for a total of 100 time steps with a period of 1 second. In this scenario, there are two circular groups of targets that are each following a linear motion model, with trajectories such that the two target groups are initially separate, with four sensors tracking each group, followed by a period with significant overlap between the two groups, with a final period of separation again. The maximum detection range of the sensors is set to $\rho_{\max}=100$\,m, with no angular FoV restrictions, $p_{D,\mathrm{max}}=0.99$, $k_{\rho}=0.5$, and the maximum communication range between sensors is set at 300\,m. In the second scenario, each sensor has 9 possible actions: remain stationary, move 15\,m east, move 15\,m north-east, move 15\,m north, move 15\,m north-west, move 15\,m west, move 15\,m south-west, move 15\,m south, and move 15\,m south-east.

\begin{figure}[t]
   \centering
   \includegraphics[width=0.8\columnwidth]{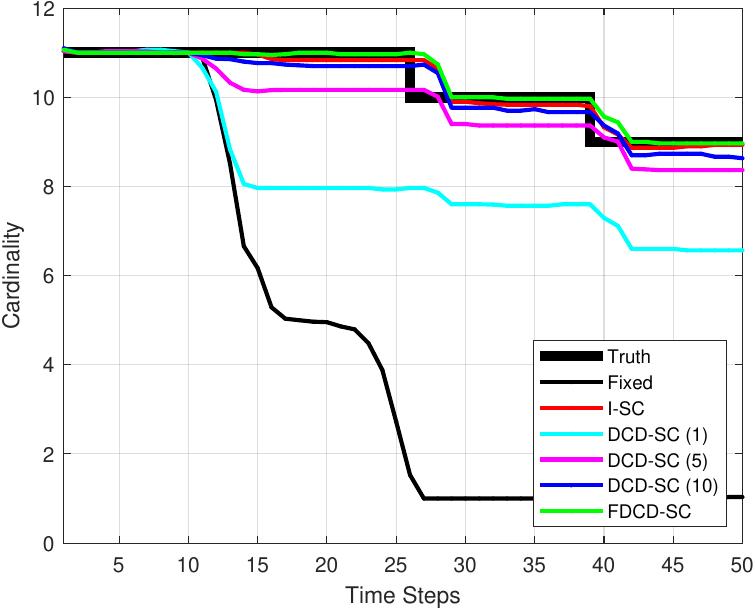}
   \caption{The average cardinality error of the fixed sensors, I-SC, DCD-SC, and FDCD-SC methods for the first scenario. The true number of targets is indicated by the thin black line, and the average estimated numbers of targets are shown by the colored lines.}
   \label{fig:6cams_cardinality_error}
\end{figure}

\begin{figure}[t]
   \centering
   \includegraphics[width=0.8\columnwidth]{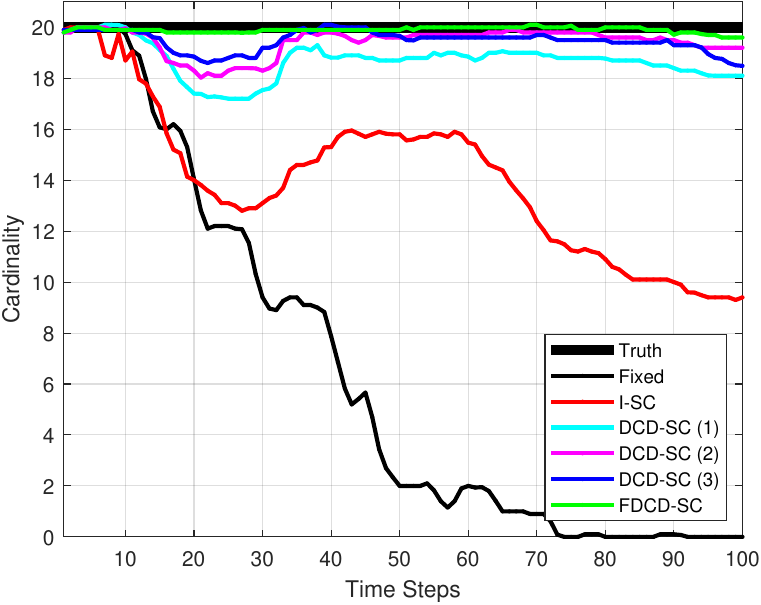}
   \caption{The average cardinality error of the fixed sensors, I-SC, DCD-SC, and FDCD-SC methods for the second scenario.}
   \label{fig:8cams_cardinality_error}
\end{figure}

\begin{table}[t]
    \centering
    \caption{Average multi-object tracking accuracy in scenario 1 of fixed sensors, I-SC, DCD-SC, and FDCD-SC.}
    \begin{tabular}{l|ccc}
    	\small
        \small\textbf{Control method} & \small\textbf{OSPA\,(m)} & \small\textbf{OSPA$^{(2)}$\,(m)} & \small\textbf{Time\,(s)} \\
        \hline
        \small Fixed Sensors & \small 60.36 & \small 48.97 & \small -- \\
        \small I-SC & \small 3.94 & \small 5.35 & \small 0.06 \\
        \small DCD-SC (1 run) & \small 22.39 & \small 21.34 & \small 0.10 \\
        \small DCD-SC (5 runs) & \small 8.36 & \small 8.65 & \small 0.43 \\
        \small DCD-SC (10 runs) & \small 5.54 & \small 7.34 & \small 0.83 \\
        \small \textbf{FDCD-SC} & \small \textbf{3.39} & \small \textbf{4.86} & \small \textbf{0.74} \\
    \end{tabular}
    \label{tab:6_sensor_results}
\end{table}

\begin{table}[t]
    \centering
    \caption{Average multi-object tracking accuracy in scenario 2 of fixed sensors, I-SC, DCD-SC, and FDCD-SC.}
    \begin{tabular}{l|ccc}
    	\small
        \small\textbf{Control method} & \small\textbf{OSPA\,(m)} & \small\textbf{OSPA$^{(2)}$\,(m)} & \small\textbf{Time\,(s)} \\
        \hline
        \small Fixed Sensors & \small 68.58 & \small 70.76 & \small -- \\
        \small I-SC & \small 31.82 & \small 46.44 & \small 0.54 \\
        \small DCD-SC (1 run) & \small 8.13 & \small 23.40 & \small 3.31 \\
        \small DCD-SC (2 runs) & \small 5.12 & \small 20.85 & \small 7.44 \\
        \small DCD-SC (3 runs) & \small 4.16 & \small 13.35 & \small 13.75 \\
        \small \textbf{FDCD-SC} & \small \textbf{1.73} & \small \textbf{5.87} & \small \textbf{12.80} \\
    \end{tabular}
    \label{tab:8_sensor_results}
\end{table}

\subsection{Methods}
The proposed Fully Distributed Coordinate Descent method (FDCD-SC) for multi-sensor control, with the constraint parameters from~\eqref{eq:final_objective_fun} set to $\psi_{\text{th}}=0.8$ and $\eta_{\text{th}}=50\,\text{m}$, $\varepsilon$ from~\eqref{eq:epsilon} set to $1e-6$ and $\lambda$ from~\eqref{eq:penalty} set to $100$, is evaluated and compared against several other approaches. Sensors with fixed positions and orientations are examined to establish baseline results. An individual sensor control method, denoted I-SC, where each sensor independently finds the optimal action for itself without considering the other sensor's actions, is evaluated to investigate the benefit of performing cooperative multi-sensor control. For the I-SC approach, the fused pseudo-updated posterior in the objective function~\eqref{eq:final_objective_fun} is replaced by the local pseudo-updated posterior, and each sensor independently performs an exhaustive search over the possible single-sensor control commands.

Finally, the recently proposed local Distributed Coordinate Descent approach (DCD-SC) is also tested to determine the comparative multi-object tracking accuracy and computation time. Wang et al.~\cite{ref:Xiaoying} describe the relationship between the number of coordinate descent runs required to reach a likelihood of finding the global optimum. Given $\mathfrak{M}$ local optima and the probability of finding the global optimum $P_{\mathrm{success}}$, the number of required runs is:

\begin{equation}
    m=\left\lceil\frac{\mathrm{log}(1-P_{\mathrm{success}})}{\mathrm{log}(1-\frac{1}{\mathfrak{M}})}\right\rceil
    \label{eq:coordinate_descent}
\end{equation}

Assuming that $\mathfrak{M}=2\lvert N_{\text{max}}\rvert$, where $N_{\text{max}}=\underset{s\in\mathbb{S}}{\text{max}}\lvert N^{(s)}\rvert$, then to achieve a $95\%$ probability of finding the global optimum in the first scenario, the sensor with the most neighbors requires 35 runs, and in the second scenario, the sensor with the most neighbors requires 47 runs. This is an impractical number of runs given the high computational cost; therefore, we aim for a sub-optimal but computationally tractable solution. To compare the computation time and accuracy trade-off, several different numbers of runs are tested. All of the multi-sensor control methods use the objective function and adaptive complementary fusion proposed in this paper.

All experiments were performed on a computer running MATLAB 2023a on Windows 10, with an AMD Ryzen 7 7700X @ 4.5GHz processor and 32 GB of RAM, to ensure consistency of computational performance across all methods. We provide our implementation and results on GitHub.\footnote{Available at: \url{https://github.com/AidanBlair/FDCD-SC}}

\subsection{Results}
The OSPA~\cite{ref:OSPA} and OSPA$^{(2)}$~\cite{ref:OSPA2} distance metrics are used to quantify the multi-target tracking accuracy of the different multi-sensor control methods. OSPA is a metric that measures both the localization and cardinality errors between the true and estimated multi-target states at each time step. The OSPA cutoff and order are set to $c=100$ and $p=1$. OSPA$^{(2)}$ extends the OSPA metric to also measure track label accuracy, accounting for track consistency. The OSPA$^{(2)}$ window length is set to $w=10$. To ensure statistical reliability, 30 Monte Carlo runs were performed for each method and each scenario.

Figure~\ref{fig:6cams_cardinality_error} shows the estimated target cardinality (number of estimated targets) averaged across the Monte Carlo runs at each time step, for each of the four tested methods in the first scenario. The closer the thin colored lines (estimated target cardinality) are to the thick black line (true target cardinality), the more accurate the number of estimated targets and, therefore, multi-object tracking accuracy.

Table~\ref{tab:6_sensor_results} details the average OSPA and OSPA$^{(2)}$ distances across all time steps and the 30 Monte Carlo runs for each method, as well as the average computation time per sensor for each method. Figure~\ref{fig:8cams_cardinality_error} and Table~\ref{tab:8_sensor_results} detail the same results for the second scenario.

\subsection{Discussion}
As discussed previously, the main source of multi-object tracking error in a scenario with a large number of mobile, dynamic targets is the targets not being observed by any of the sensors, i.e., cardinality error. For the first scenario, Figure~\ref{fig:6cams_cardinality_error} and Table~\ref{tab:6_sensor_results} show how, when the sensors are fixed, while initially the targets are all within the FoV of the sensor network, when the targets begin to move outside of the network's FoV the targets are rapidly lost, leading to a high cardinality error and poor multi-object tracking performance. In comparison, when the various multi-sensor control methods are used, the sensors are able to follow the trajectories of the targets, leading to a much lower cardinality error.


The proposed FDCD-SC method has the best cardinality error of all of the evaluated methods, rarely dropping targets due to subsequent missed detections, with the majority of the cardinality inaccuracies occurring at target deaths. The I-SC method performs similarly to FDCD-SC, with only a slightly worse cardinality error. This is due to the bottom four sensors all being able to track disjoint groups of 2 out of the 8 targets around them without the need for coordination; however, around time step 15, one of the targets starting at $(0, 300)$ or $(0, 700)$ tends to get dropped, due to the lack of communication between the two top sensors. DCD-SC with $m=1$ performs poorly, typically dropping 3 targets starting from time step 10. However, as $m$ increases, the cardinality error decreases, approaching the errors of ISC and FDCD-SC at $m=10$. This suggests that a large portion of the cardinality error at lower $m$ values is due to poor optimization of the local multi-sensor control commands, while the smaller portion that remains at any $m$ is due to the lack of global coordination between the sensors, resulting in sensors taking actions that drop targets.

Table~\ref{tab:6_sensor_results} highlights the significant improvement that multi-sensor control has over fixed, stationary sensors when tracking multiple dynamic targets, as well as the reduction in both the OSPA and OSPA$^{(2)}$ distances in FDCD-SC compared to I-SC and particularly DCD-SC. The lower multi-object tracking errors of FDCD-SC compared to ISC come at the cost of a significantly longer computation time than I-SC (\raisebox{0.5ex}{\texttildelow}12 times longer). By the time $m=10$, the computation time of DCD-SC has already surpassed FDCD-SC, suggesting that there is no accuracy-computation tradeoff between the two methods.



Figure~\ref{fig:8cams_cardinality_error} shows similar results in the more challenging second scenario. The proposed FDCD-SC method again has the lowest cardinality error of all of the evaluated methods, with a very low number of targets being dropped, demonstrating its capacity to consistently find globally optimal multi-sensor control commands. I-SC performs much worse in this scenario, quickly dropping a large number of targets, resulting in poor cardinality accuracy. This is caused by the arrangement of the targets and the small sensor FoV, leading to many instances of locally optimal single-sensor actions resulting in globally suboptimal multi-sensor actions. The DCD-SC method follows an interesting pattern in this scenario that can be divided into three periods. Time steps $1-20$, time steps $21-85$, and time steps $86-100$. From time step $k=1$ to $k\approx20$, the sensors are spread out far enough that there is no sensor node $s$ where $N^{(s)}=\mathcal{S}$, so there is a lack of global coordination and targets are dropped. However, in the second period from $k\approx21$ to $k\approx85$, targets are rarely dropped and the cardinality error tends to improve, as the sensors have moved close enough together such that at least one sensor is within communication range of the entire network, resulting in DCD-SC functioning identically to FDCD-SC and selecting globally optimal multi-sensor control commands. Finally, in the third period of $k\approx86$ to $k=100$, the sensors move far enough apart such that no sensor can communicate directly with the entirety of the network, and targets are again quickly dropped and cardinality error worsens. This behavior further supports the idea that the poor cardinality accuracy in DCD-SC is caused by the lack of global coordination, rather than suboptimal local optimization. Table~\ref{tab:8_sensor_results} highlights that FDCD-SC has the lowest OSPA and OSPA$^{(2)}$ errors of all the methods tested. The computation time for FDCD-SC is again significantly longer in comparison to I-SC (\raisebox{0.5ex}{\texttildelow}24 times longer) and performs better than DCD-SC at a comparable computation time ($m=3$).

We then compare the computational complexity of each evaluated method for a single time step. For all cases, the number of targets, the dimensionality of the state space, and the number of particles are assumed to be constant and are thus not included. The computational complexity of ACF in the worst case ($\mathcal{A}_k^{(\ell)}=\mathbb{S}$ for all $\ell$) scales linearly with the number of sensors, $\mathcal{O}(N_{\text{max}})$ for DCD-SC and $\mathcal{O}(\mathcal{S})$ for FDCD-SC. The computational cost for ISC scales linearly with the number of possible actions, $\mathcal{O}\left(\lvert\mathbb{U}\rvert\right)$. The computational cost for DCD-SC scales with the number of possible actions, the maximum local sensor neighborhood size, the number of coordinate descent iterations $I$, and the number of runs $m$ (with fusion occurring once at each multi-sensor control command evaluation), and is therefore $\mathcal{O}\left(mIN_{\text{max}}^2\lvert\mathbb{U}\rvert\right)$. The computational cost for FDCD-SC scales with the number of possible actions, the number of sensors, and the number of coordinate descent iterations (with fusion occurring once at each multi-sensor control command evaluation) and is therefore $\mathcal{O}\left(\lvert I\mathcal{S}^2\mathbb{U}\rvert\right)$.

For the same scenario, FDCD-SC is $I\mathcal{S}^2$ times more expensive than ISC and $\frac{\mathcal{S}^2}{mN_{\text{max}}^2}$ times more expensive than DCD-SC. If $m\left(\frac{N_{\text{max}}}{\mathcal{S}}\right)^2>1$, which is more likely in densely connected networks, then FDCD-SC will be computationally faster than DCD-SC. Conversely, if $m\left(\frac{N_{\text{max}}}{\mathcal{S}}\right)^2<1$, which is more likely in sparsely connected networks, then FDCD-SC will be computationally slower than DCD-SC.

\section{Conclusion}
\label{sec:Conclusion}
This paper addresses multi-sensor control in a distributed sensor network with a focus on multi-target tracking applications. First, an adaptive complementary fusion rule for LMB densities with advantageous properties for multi-sensor control was developed. Then, an information-theoretic objective function with collision-avoidance constraints was proposed. Finally, a fully distributed coordinate descent multi-sensor control method (FDCD-SC) was proposed and evaluated against several baseline methods. To evaluate the tracking performance of these methods, two scenarios with varying sensor \& target numbers, sensor profiles, and actions were designed, and the OSPA and OSPA$^{(2)}$ error metrics were employed.

The results demonstrated that FDCD-SC significantly outperformed all compared methods in multi-object tracking accuracy in both scenarios, both in terms of cardinality error and OSPA \& OSPA$^{(2)}$ errors, demonstrating the advantages of performing globally coordinated multi-sensor control in multi-object tracking. However, FDCD-SC has a higher computational cost than I-SC, and depending on the sensor network topology and the number of runs that are performed with DCD-SC also has a higher computational cost than DCD-SC. For future research, a distributed multi-sensor control that is robust to communication dropout/attacks could be explored.

\section*{Acknowledgment}

The Australian Research Council supported this work through Grant DE210101181.

\bibliographystyle{IEEEtran}
\bibliography{AidanBib.bib}

\end{document}